\documentclass{article}

\usepackage[a4paper,hmargin=2.5cm,vmargin=3cm]{geometry}
\usepackage{microtype}                  
\usepackage{graphicx}                   
\usepackage{amsmath}                    
\usepackage{amssymb}                    
\usepackage{amsthm}                     
\usepackage{amsfonts}                   
\usepackage{url}                        
\usepackage{color}                      
\usepackage{cite}                       
\usepackage{hyperref}                   
\hypersetup{colorlinks=true,linkcolor=[rgb]{0.75,0,0},citecolor=[rgb]{0,0,0.75}}%
\usepackage{wrapfig}                    

\newtheorem{theorem}{Theorem}[]
\newtheorem{lemma}[theorem]{Lemma}

\newcommand{\floor}[1]{\left\lfloor #1\right\rfloor}

\newcommand{\RE}{\mathbb{R}}            
\newcommand{\eps}{\varepsilon}          

\newcommand{\inv}[1]{\frac{1}{#1}}
\DeclareMathOperator{\conv}{conv}
\DeclareMathOperator{\radius}{radius}
\DeclareMathOperator{\width}{width}
\DeclareMathOperator{\thick}{thick}

\newcommand{\etal}{\textit{et al.}}

\DeclareMathOperator{\polylog}{polylog}
\newcommand{\plog}{\polylog\inv\eps}

\begin{document}

\title{Approximate Convex Intersection Detection with Applications to Width and Minkowski Sums}

\author{%
	Sunil Arya\thanks{Research supported by the Research Grants Council of Hong Kong, China under project number 16200014.}\\
		Department of Computer Science and Engineering \\
		The Hong Kong University of 
        Science and Technology \\
 		Clear Water Bay, Kowloon, Hong Kong\\
		arya@cse.ust.hk \\
		\and
	Guilherme D. da Fonseca\thanks{Research supported by the European Research Council under ERC Grant Agreement number 339025 GUDHI (Algorithmic Foundations of Geometric Understanding in Higher Dimensions).}\\
		Universit\'{e} Clermont Auvergne,\\
		LIMOS, and INRIA Sophia Antipolis\\
 		France\\
		fonseca@isima.fr
		\and
	David M. Mount\thanks{Research supported by NSF grant CCF--1618866.}\\
		Department of Computer Science and \\
		Institute for Advanced Computer Studies \\
		University of Maryland \\
 		College Park, Maryland 20742 \\
		mount@cs.umd.edu \\
}

\date{}

\maketitle

\begin{abstract}
Approximation problems involving a single convex body in $\mathbb{R}^d$ have received a great deal of attention in the computational geometry community. In contrast, works involving multiple convex bodies are generally limited to dimensions $d \leq 3$ and/or do not consider approximation. In this paper, we consider approximations to two natural problems involving multiple convex bodies: detecting whether two polytopes intersect and computing their Minkowski sum. Given an approximation parameter $\varepsilon > 0$, we show how to independently preprocess two polytopes $A,B \subset \mathbb{R}^d$ into data structures of size $O(1/\varepsilon^{(d-1)/2})$ such that we can answer in polylogarithmic time whether $A$ and $B$ intersect approximately. More generally, we can answer this for the images of $A$ and $B$ under affine transformations. Next, we show how to $\varepsilon$-approximate the Minkowski sum of two given polytopes defined as the intersection of $n$ halfspaces in $O(n \log(1/\varepsilon) + 1/\varepsilon^{(d-1)/2 + \alpha})$ time, for any constant $\alpha > 0$. Finally, we present a surprising impact of these results to a well studied problem that considers a single convex body. We show how to $\varepsilon$-approximate the width of a set of $n$ points in $O(n \log(1/\varepsilon) + 1/\varepsilon^{(d-1)/2 + \alpha})$ time, for any constant $\alpha > 0$, a major improvement over the previous bound of roughly $O(n + 1/\varepsilon^{d-1})$ time.
\end{abstract}


\section{Introduction} \label{s:intro}

Approximation problems involving a single convex body in $d$-dimensional space have received a great deal of attention in the computational geometry community~\cite{AHV04,AFM17b,AFM17c,AFM17a,AFM18a,Cha06,Cha17,YAP08}. Recent results include near-optimal algorithms for approximating the convex hull of a set of points~\cite{AFM17b,Cha17}, as well as an optimal data structure for answering approximate polytope membership queries~\cite{AFM17a}.
In contrast, works involving multiple convex bodies are generally limited to dimensions $d \leq 3$ and/or do not consider approximation~\cite{AGHRS00,BaL15,FHW90,GXG08,VaM06}. In this paper we present new approximation algorithms to natural problems that either involve multiple convex polytopes or result from such an analysis:
\begin{itemize}
\item Determining whether two convex polytopes $A$ and $B$ intersect

\item Computing the Minkowski sum, $A \oplus B$, of two convex polytopes

\item Computing the width of a convex polytope $A$ (which results from an analysis of the Minkowski sum $A \oplus (-A)$)
\end{itemize}

Throughout we assume that the input polytopes reside in $\RE^d$ and are full-dimensional, where the dimension $d$ is a fixed constant. Polytopes may be represented either as the convex hull of $n$ points (\emph{point representation}) or as the intersection of $n$ halfspaces (\emph{halfspace representation}). In either case, $n$ denotes the \emph{size} of the polytope.

\subsection{Convex Intersection} \label{s:intersection-results}

Detecting whether two geometric objects intersect and computing the region of intersection are fundamental problems in computational geometry. Geometric intersection problems arise naturally in a number of applications. Examples include geometric packing and covering, wire and component layout in VLSI, map overlay in geographic information systems, motion planning, and collision detection. Several surveys present the topics of collision detection and geometric intersection~\cite{JTT01,LiG98,Mou17}.

The special case of detecting the intersection of convex objects has received a lot of attention in computational geometry. The static version of the problem has been considered in $\RE^2$~\cite{oRo98,Sha75} and $\RE^3$~\cite{Cha92,MuP78}. The data structure version where each convex object is preprocessed independently has been considered in $\RE^2$~\cite{BaL15,ChD80,ChD87,DoK83} and $\RE^3$~\cite{BaL15,ChD87,DoK83,DoK90}.

Recently, Barba and Langerman~\cite{BaL15} considered the problem in higher dimension. They showed how to preprocess convex polytopes in $\RE^d$ so that given two such polytopes that have been subject to affine transformations, it can be determined whether they intersect each other in logarithmic time. However, the preprocessing time and storage grow as the combinatorial complexity of the polytope raised to the power $\floor{d/2}$. Since the combinatorial complexity of a polytope with $n$ vertices can be as high as $\Theta(n^{\floor{d/2}})$, the storage upper bound is roughly $O(n^{d^2/4})$. This high complexity motivates the study of approximations to the problem.

We define approximation in a manner that is sensitive to direction. Consider any convex body $K$ in $\RE^d$ and any $\eps > 0$. Given a nonzero vector $v \in \RE^d$, define $\Pi_v(K)$ to be the minimum slab defined by two hyperplanes that enclose $K$ and are orthogonal to $v$. Define the \emph{directional width} of $K$ with respect to $v$, $\width_v(K)$, to be the perpendicular distance between these hyperplanes. Let $\Pi_{v,\eps}(K)$ be the central expansion of $\Pi_v(K)$ by a factor of $1+\eps$, and define $K_\eps$ to be the intersection of these expanded slabs over all unit vectors $v$. It can be shown that for any $v$, $\width_v(K_\eps) = (1+\eps) \width_v(K)$. An \emph{$\eps$-approximation} of $K$ is any set $K'$ (which need not be convex) such that $K \subseteq K' \subseteq K_\eps$. This defines an \emph{outer} approximation. It is also possible to define an analogous notion of \emph{inner} approximation in which each directional width is no smaller than $1-\eps$ times the true width. Our results can be extended to either type of approximation.

A related notion studied extensive in the literature is that of $\eps$-kernels.
Given a discrete point set $S$ in $\RE^d$, an \emph{$\eps$-kernel} of $S$ is any subset $Q \subseteq S$ such that $\conv(Q)$ is an inner $\eps$-approximation of $\conv(S)$~\cite{AHV04}. It is well known that $O(1/\eps^{(d-1)/2})$ points are sufficient and sometimes necessary in an $\eps$-kernel.
Kernels efficiently approximate the convex hull and as such have been used to obtain fast approximation algorithms to several problems such as diameter, minimum width, convex hull volume, minimum enclosing cylinder, minimum enclosing annulus, and minimum-width cylindrical shell~\cite{AHV04,AHV05}.

In the $\eps$-approximate version of convex intersection, we are given two convex bodies $A$ and $B$ and a parameter $\eps > 0$. If $A \cap B \neq \emptyset$, then the answer is ``\emph{yes}.'' If $A_{\eps} \cap B_{\eps} = \emptyset$, then the answer is ``\emph{no}.'' Otherwise, either answer is acceptable. The \emph{$\eps$-approximate polytope intersection problem} is defined as follows. A collection of two or more convex polytopes in $\RE^d$ are individually preprocessed (with knowledge of $\eps$). Given any two preprocessed polytopes, $A$ and $B$, the query determines whether $A$ and $B$ intersect approximately. In general, the query algorithm can be applied to any affine transformation of the preprocessed polytopes.

\begin{theorem} \label{thm:intersection}
Given a parameter $\eps > 0$ and two polytopes $A,B \subset \RE^d$ each of size $n$ (given either using a point or halfspace representation), we can independently preprocess each polytope into a data structure in order to answer $\eps$-approximate polytope intersection queries with query time $O(\plog)$, storage $O(1/\eps^{(d-1)/2})$, and preprocessing time
$O(n \log\inv\eps + 1/\eps^{(d-1)/2+\alpha})$, where $\alpha$ is an arbitrarily small positive constant.
\end{theorem}

The space is nearly optimal in the worst case because there is a lower bound of $\Omega(1/\eps^{(d-1)/2})$ on the worst-case bit complexity of representing an $\eps$-approximation of a polytope~\cite{AFM17a}.

\subsection{Minkowski Sum}

Given two convex bodies $A,B \subset \RE^d$, the \emph{Minkowski sum} $A \oplus B$ is defined as $\{p+q : p \in A,\; q \in B\}$ (see Figure~\ref{f:minkowski}(a)). Minkowski sums have found numerous applications in motion planning~\cite{ArS97,HSS17}, computer-aided design~\cite{VaM06}, computational biology~\cite{PaS05}, satellite layout~\cite{BDT97}, and image processing~\cite{KaR92}. Minkowski sums have also been well studied in the context of discrete and computational geometry~\cite{AFH02,AHKS14,FHW90,HCAHS95,Tiw08}. 

\begin{figure}[tbp]
  \centerline{\includegraphics[scale=0.7]{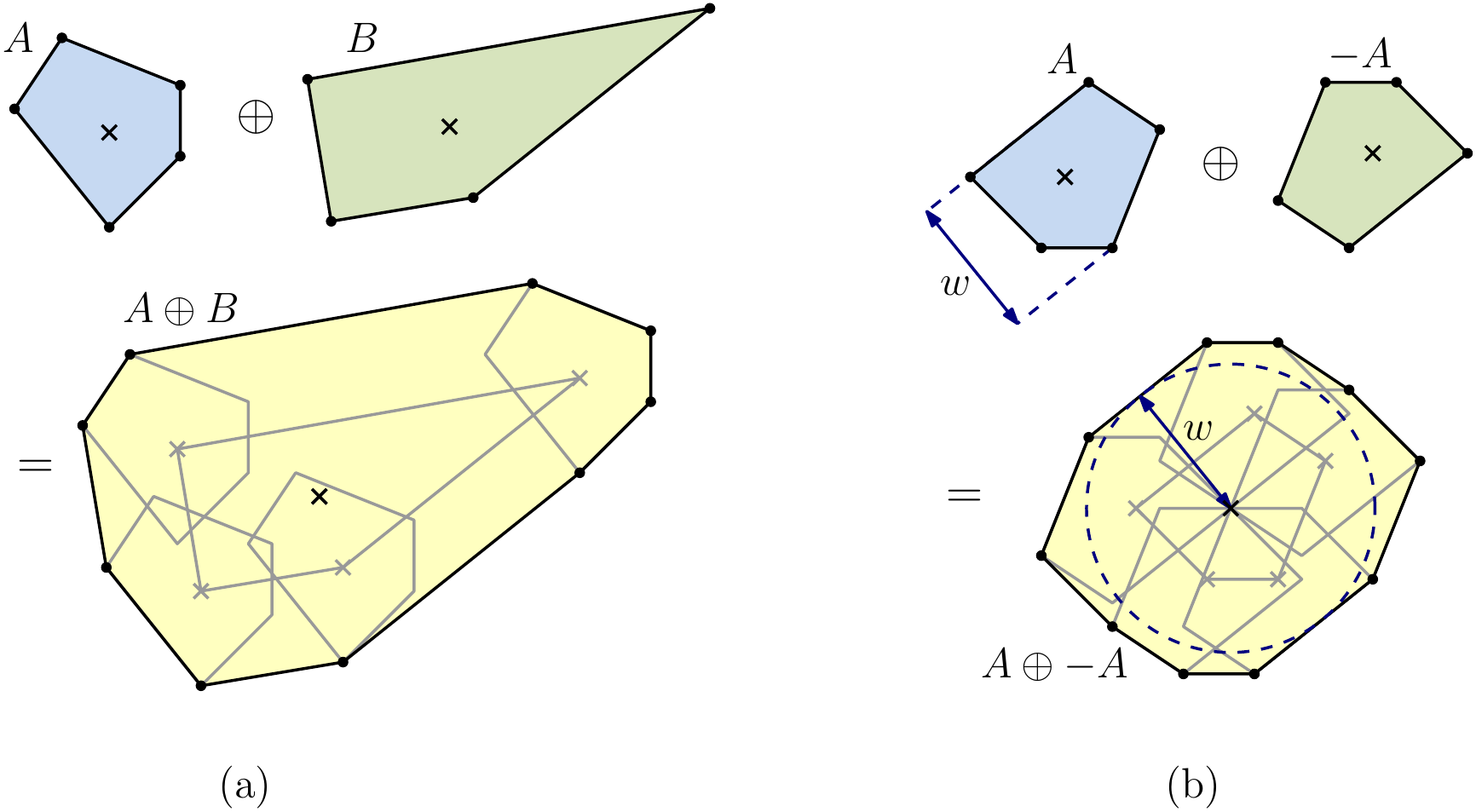}}
  \caption{Minkowski sum and its relationship to width.}
  \label{f:minkowski}
\end{figure}

It is well known that in dimension $d \ge 3$, the number of vertices in the Minkowski sum of two polytopes can grow as rapidly as the product of the number of vertices in the two polytopes \cite{ArS97}. This has led to the study of algorithms to compute approximations to Minkowski sums in $\RE^3$~\cite{AGHRS00,GXG08,VaM06}. In this paper, we show how to approximate the Minkowski sum of two convex polytopes in $\RE^d$ in near-optimal time.

\begin{theorem} \label{thm:minkowski}
Given a parameter $\eps > 0$ and two polytopes $A,B \subset \RE^d$ each of size $n$ (given either using a point or halfspace representation), it is possible to construct an $\eps$-approximation of $A \oplus B$ of size $O(1/\eps^{(d-1)/2})$ in $O(n \log\inv\eps + 1/\eps^{(d-1)/2+\alpha})$ time, where $\alpha$ is an arbitrarily small positive constant.
\end{theorem}

The output representation can be either point-based or halfspace-based, irrespective of the input representations. 

\subsection{Width}

Define the directional width of a set $S$ of $n$ points to be the directional width of $\conv(S)$. The \emph{width} of $S$ is the \emph{minimum} over all directional widths. The \emph{maximum} over all directional widths is equal to the diameter of $S$.
Both problems can be approximated using the $\eps$-kernel of $S$.
After successive improvements~\cite{AHV04,AMS92,ArC14,BaH01,Cha06}, algorithms to compute $\eps$-kernels and to $\eps$-approximate the diameter in roughly $O(n + 1/\eps^{d/2})$ time have been independently discovered by Chan~\cite{Cha17} and the authors~\cite{AFM17b}.
Somewhat surprisingly, these works offer no improvement to the running time to approximate the width~\cite{AHV04,Cha02b,Cha06,DGR97,YAP08}, which Chan~\cite{Cha17} posed as an open problem. The fastest known algorithms date from over a decade ago and take roughly $O(n + 1/\eps^{d-1})$ time~\cite{Cha02b,Cha06}.

Agarwal \etal~\cite{AGHRS00} showed that the width of a convex body $K$ is equal to the minimum distance from the origin to the boundary of the convex body $K \oplus (-K)$ (see Figure~\ref{f:minkowski}(b)). Using Theorem~\ref{thm:minkowski}, we can approximate the width by computing an $\eps$-approximation of $K \oplus (-K)$ represented as the intersection of halfspaces and then determining the closest point to the origin among all bounding hyperplanes. The following presents this result.

\begin{theorem} \label{thm:width}
Given a set $S$ of $n$ points in $\RE^d$ and an approximation parameter $\eps > 0$, it is possible to compute an $\eps$-approximation to the width of $S$ in $O(n \log\inv\eps + 1/\eps^{(d-1)/2+\alpha})$ time, where $\alpha$ is an arbitrarily small positive constant.
\end{theorem}

\subsection{Techniques}

Our algorithms and data structure are based on a data structure defined by a hierarchy of Macbeath regions~\cite{AFM17b,AFM17a}, which answers approximate directional width queries in polylogarithmic time. First, we show how to use this data structure as a black box to answer approximate polytope intersection queries by transforming the problem to a dual setting and performing a multidimensional convex minimization. Next, we show how to use approximate polytope intersection queries to compute $\eps$-approximations of the Minkowski sum. The approximation to the width follows directly.

Since we only access the input polytopes through a data structure for approximate directional width queries, our results apply in much more general settings. For example, we could answer in polylogarithmic time whether the Minkowski sum of two polytopes (preprocessed independently) approximately intersects a third polytope. Our techniques are also amenable to other polytope operations such as intersection and convex hull of the union, as long as the model of approximation is defined accordingly.

The preprocessing time of the approximate directional width data structure we use is $O(n \log\inv\eps + 1/\eps^{(d-1)/2+\alpha})$, for arbitrarily small $\alpha > 0$. If this preprocessing time is reduced in the future, the complexity of our algorithms becomes equal to the preprocessing time plus $O((1 / \eps^{(d-1)/2}) \plog)$.

\section{Preliminaries} \label{s:preliminaries}

In this section we present a number of results, which will be used throughout the paper. The first provides three basic properties of Minkowski sums. The proof can be found in standard sources on Minkowski sums (see, e.g.,~\cite{Sch93}).

\begin{lemma} \label{lem:basic}
Let $A,B \subset \RE^d$ be two (possibly infinite) sets of points. Then:
\begin{enumerate}
 \item[$(a)$\hspace{-3pt}] $A \cap B \neq \emptyset$ if and only if $O \in A \oplus (-B)$, where $O$ is the origin.

 \item[$(b)$\hspace{-3pt}] $\conv(A \oplus B) = \conv(A) \oplus \conv(B)$.

 \item[$(c)$\hspace{-3pt}] For all nonzero vectors $v$, $\width_v(A \oplus B) = \width_v(A) + \width_v(B)$.
\end{enumerate}
\end{lemma}

Next, we recall a recent result of ours on answering directional width queries approximately~\cite{AFM17b}, which we will use as a black box later in this paper. Given a set $S$ of $n$ points in a constant dimension $d$ and an approximation parameter $\eps>0$, the answer to the \emph{approximate directional width query} for a nonzero query vector $v$ consists of a pair of points $p,q \in S$ such that $\width_v(\{p,q\}) \geq (1-\eps) \; \width_v(S)$.

\begin{lemma} \label{lem:widthqueries}
Given a set $S$ of $n$ points in $\RE^d$ and an approximation parameter $\eps > 0$, there is a data structure that can answer $\eps$-approximate directional width queries with
query time $O( \log^2 \inv{\eps} )$,
space $O(1/\eps^{(d-1)/2})$, and
preprocessing time $O( n \log \inv \eps + 1/\eps^{(d-1)/2+\alpha})$.
\end{lemma}

\subsection{Fattening} \label{ss:fattening}
 
Existing algorithms and data structures for convex approximation often assume that the bodies have been fattened through an appropriate affine transformation. In the context of multiple bodies, this is complicated by the fact that different fattening transformations may be needed for the two bodies or their Minkowski sum. In this section we explore this issue.

Consider a convex body $K$ in $d$-dimensional space $\RE^d$. Given a parameter $0 < \gamma \le 1$, we say that $K$ is \emph{$\gamma$-fat} if there exist concentric Euclidean balls $B$ and $B'$, such that $B \subseteq K \subseteq B'$, and $\radius(B) / \radius(B') \ge \gamma$. We say that $K$ is \emph{fat} if it is $\gamma$-fat for a constant $\gamma$ (possibly depending on $d$, but not on $\eps$ or $K$). For a centrally symmetric convex body $C$, the body obtained by scaling $C$ about its center by a factor of $\lambda$ is called the $\lambda$-expansion of $C$.

Let $K$ be a convex body. We say that a convex body $C$ is a \emph{$\lambda$-sandwiching} body for $K$ if $C$ is centrally symmetric and $C \subseteq K \subseteq C'$, where $C'$ is a $\lambda$-expansion of $C$. John~\cite{Joh48} proved tight bounds for the constant $\lambda$ of a $\lambda$-sandwiching ellipsoid. This ellipsoid is referred to as the \emph{John ellipsoid}.

\begin{lemma} \label{lem:John}
For every convex body $K$ in $\RE^d$, there exists a $d$-sandwiching ellipsoid. Furthermore, if $K$ is centrally symmetric, there exists a $\sqrt{d}$-sandwiching ellipsoid.
\end{lemma}

It is an immediate consequence of this lemma that for any convex body $K$ there exists an affine transformation $T$ such that $T(K)$ is $(1/d)$-fat. Any affine transformation that maps the John ellipsoid into a Euclidean ball will do. The following lemma generalizes this to hyperrectangles (see also Barequet and Har-Peled~\cite{BaH01}).

\begin{lemma} \label{lem:rect} 
For every convex body $K$ in $\RE^d$, there exists a $(d^{3/2})$-sandwiching hyperrectangle.
\end{lemma}

\begin{proof}
Let $E$ denote the $d$-sandwiching ellipsoid for $K$, described in Lemma~\ref{lem:John}. By elementary geometry, there exists a $\sqrt{d}$-sandwiching hyperrectangle $R$ for $E$. We claim that $R$ is a $(d^{3/2})$-sandwiching hyperrectangle for $K$. To prove this claim, observe that $R \subseteq E \subseteq R'$ and $E \subseteq K \subseteq E'$, where $R'$ is the $\sqrt{d}$-expansion of $R$ and $E'$ is the $d$-expansion of $E$. Letting $R''$ denote the $d$-expansion of $R'$, it is easy to see that $E' \subseteq R''$. It follows that $R \subseteq E \subseteq K \subseteq E' \subseteq R''$. Since $R''$ is the $d$-expansion of $R'$ and $R'$ is the $\sqrt{d}$-expansion of $R$, it follows that $R''$ is the $(d^{3/2})$-expansion of $R$. This completes the proof.
\end{proof}

Next, let us consider fattening in the context of multiple bodies. The next two lemmas follow from elementary geometry and properties of Minkowski sums.

\begin{lemma} \label{lem:combine}
Let $C_1$ and $C_2$ be $\lambda$-sandwiching bodies for $K_1$ and $K_2$, respectively. Then $C_1 \oplus C_2$ is a $\lambda$-sandwiching body for $K_1 \oplus K_2$.
\end{lemma}

\begin{lemma} \label{lem:affine}
Let $K$ be a convex body. Given a $\lambda$-sandwiching polytope for $K$ of constant complexity, we can compute a $\gamma$-fattening affine transformation $T$ for $K$ in constant time, where $\gamma = 1/(\lambda \sqrt{d})$.
\end{lemma}

\begin{proof}
Let $C$ denote the given $\lambda$-sandwiching polytope for $K$.
Recalling that $\lambda$-sandwiching polytopes are centrally symmetric, by Lemma~\ref{lem:John} we can find a $\sqrt{d}$-sandwiching ellipsoid $E$ for $C$. As $C$ has constant complexity, we can determine $E$ in $O(1)$ time. In $O(1)$ time, we can also find the affine transformation $T$ that converts $E$ into a Euclidean ball. We claim that $T(K)$ is $\gamma$-fat for $\gamma = 1/(\lambda \sqrt{d})$. To prove this claim, observe that $E \subseteq C \subseteq E'$ and $C \subseteq K \subseteq C'$, where $E'$ is the $\sqrt{d}$-expansion of $E$ and $C'$ is the $\lambda$-expansion of $C$. Letting $E''$ denote the $\lambda$-expansion of $E'$, it is easy to see that $C' \subseteq E''$. It follows that $E \subseteq C \subseteq K \subseteq C' \subseteq E''$. Since $E'$ is the $\sqrt{d}$-expansion of $E$ and $E''$ is the $\lambda$-expansion of $E'$, it follows that $E''$ is the $\lambda \sqrt{d}$-expansion of $E$. Thus $T(K)$ is contained between Euclidean balls $T(E)$ and $T(E'')$, whose radii differ by a factor of $\lambda \sqrt{d}$, which proves the lemma.
\end{proof}

We conclude by showing that we can maintain a small amount of auxiliary information for any collection of convex bodies in order to determine the fattening transformation for the Minkowski sum of any two members of this library. We refer to the data structure for approximate directional width queries from Lemma~\ref{lem:widthqueries} together with the additional information to determine the fattening transformation as the \emph{augmented data structure} for approximate directional width queries.

\begin{lemma} \label{lem:fatten}
Consider any finite collection of convex polytopes in $\RE^d$, and let $\gamma = 1/d^2$. It is possible to store information of constant size with each polytope such that in constant time we can compute a $\gamma$-fattening affine transformation for the Minkowski sum of any two polytopes from the collection. This information can be computed in time proportional to the size of the input polytope.
\end{lemma}

\begin{proof}
At preprocessing time, we store the $\lambda$-sandwiching hyperrectangles $R_i$ for each $K_i$, where $\lambda = d^{3/2}$. By Lemma~\ref{lem:rect}, such hyperrectangles exist and they can be computed in time proportional to the size of the input polytope~\cite{ChM96}.

Suppose we want to compute a $\gamma$-fattening affine transformation for $K'_i \oplus K'_j$, where $K'_i$ and $K'_j$ are the result of applying (possibly different) affine transformations to $K_i$ and $K_j$, respectively. Let $C'_i$ and $C'_j$ be the polytopes of constant complexity obtained by applying the corresponding affine transformations to $R_i$ and $R_j$, respectively. Clearly, $C'_i$ and $C'_j$ are $\lambda$-sandwiching polytopes for $K'_i$ and $K'_j$, respectively. Thus, by Lemma~\ref{lem:combine}, $C'_i \oplus C'_j$ is a $\lambda$-sandwiching polytope for $K'_i \oplus K'_j$. Note that this polytope has constant complexity and can be computed in constant time. Applying Lemma~\ref{lem:affine}, we can use this polytope to compute a $\gamma$-fattening affine transformation for $K'_i \oplus K'_j$ in constant time, where $\gamma = 1/ (\lambda \sqrt{d}) = 1/d^2$. 
\end{proof}

The previous lemma holds more generally even when each of the polytopes are subject to any non-singular affine transformation and to the Minkowski sum of a constant number of polytopes.

\subsection{Projective Duality and Width} \label{ss:duality}

Our algorithm for approximating the directional width of a point set is based on a \emph{projective dual transformation}, which maps points into hyperplanes and vice versa. Each primal point $p = (p_1,\ldots,p_d) \in S$ is mapped to the dual hyperplane $p^*: x_d = p_1 x_1 + \cdots + p_{d-1}x_{d-1} - p_d$. Each primal hyperplane is mapped to a dual point in the same manner. This dual transformation has several well-known properties~\cite{textbook}. For example, the points in the lower convex hull of $S$ map to the hyperplanes in the upper envelope.

Let $H$ be a set of $n$ hyperplanes in $\RE^d$. Given a point $r \in \RE^{d-1}$, the \emph{thickness} of $H$ at $r$, denoted $\thick_r(H)$ is defined as follows. Given $r \in \RE^{d-1}$ and $t \in \RE$, let $(r,t)$ denote the point in $\RE^d$ resulting by concatenating $r$ and $t$. For the sake of illustration, we think of the $d$-th coordinate axis as being the vertical axis. Let $r' = (r,t_1)$ and $r'' = (r,t_2)$. We define $\thick_r(H)$ as the maximum difference $t_2-t_1$ for points $r',r''$ in the hyperplanes in $H$. In other words, the thickness is the vertical distance between the intersection of the vertical line defined by $r$ with the upper and lower envelopes of $H$. The following relates width and thickness.

\begin{lemma} \label{lem:dual}
Consider two points $p,q \in \RE^d$ and a vector $v = (v_1,\ldots,v_{d-1},-1)$. Let $p^*,q^*$ denote the dual hyperplanes and $v_{1,d-1} = (v_1,\ldots,v_{d-1})$. We have
\[\thick_{v_{1,d-1}}(\{p^*,q^*\}) = \|v\| \width_v(\{p,q\}).\]
\end{lemma}
\begin{proof}
Given vectors $u$ and $v$, let $u \cdot v$ denote the standard inner product. Assume without loss of generality that $p \cdot v \geq q \cdot v$. Clearly, $v$ is nonzero, so $\width_v(\{p,q\}) = (p \cdot v - q \cdot v)/\|v\|$. 
Let $p = (p_1,\ldots,p_d)$ and $q = (q_1,\ldots,q_d)$. The dual hyperplanes are 
\[
    p^*: x_d = p_1 x_1 + \cdots + p_{d-1}x_{d-1} - p_d \quad\hbox{and}\quad q^*: x_d = q_1 x_1 + \cdots + q_{d-1}x_{d-1} - q_d.
\]
If we set $x_1,\ldots,x_{d-1} = v_{1,d-1}$ we have $t_2 = (p_1,\ldots,p_{d-1}) \cdot v_{1,d-1} - p_d$ and $t_1 = (q_1,\ldots,q_{d-1}) \cdot v_{1,d-1} - q_d$. Therefore
\begin{align*}
\thick_{v_{1,d-1}}(H) & = t_2-t_1 \\
 & = (p_1,\ldots,p_{d-1}) \cdot v_{1,d-1} - p_d - ((q_1,\ldots,q_{d-1}) \cdot v_{1,d-1}  - q_d) \\
 & = p \cdot v - q \cdot v \\
 & = \|v\| \width_v(\{p,q\}). \qedhere
\end{align*}
\end{proof}

\section{Approximate Convex Intersection} \label{s:intersection}

In this section, we will prove Theorem~\ref{thm:intersection} for the case when the input polytopes are represented by points. Assume that we are given two polytopes $A$ and $B$ in the point representation. The objective is to preprocess $A$ and $B$ individually such that we can efficiently answer approximate intersection queries for $A$ and $B$ (or more generally for affine transformations of $A$ and $B$). 

Given a convex body $K$, $\eps > 0$, and a point $p$, an \emph{$\eps$-approximate polytope membership query} is defined as follows.
If $p \in K$, the answer is ``\emph{yes},'' if $p \notin K_\eps$, the answer is ``\emph{no},'' and otherwise, either answer is acceptable.
Our strategy to answer approximate intersection queries is based on reducing them to approximate polytope membership queries. This reduction is presented in the following lemma, which is a straightforward generalization of Lemma~\ref{lem:basic}(a) to an approximate context. The proof follows from standard algebraic properties of Minkowski sums and the observation that $K_\eps$ can be expressed as $K \oplus \frac{\eps}{2}(K \oplus -K)$.

\begin{lemma}
Let $A,B \subset \RE^d$ be two polytopes and $\eps > 0$. Determining the $\eps$-approximate intersection of $A$ and $B$ is equivalent to determining the $\eps$-approximate membership of $O \in A \oplus (-B)$.
\end{lemma}

\begin{proof}
We begin by establishing the useful identity $A_\eps \oplus B_\eps = (A \oplus B)_\eps$. By basic properties of Minkowski sums (commutativity and distributivity) we have
\begin{align*}
 A_\eps \oplus B_\eps 
     & = \left(A \oplus \frac{\eps}{2} (A \oplus -A)\right) \oplus \left(B \oplus \frac{\eps}{2}(B \oplus -B)\right) \\
     & = (A \oplus B) \oplus \frac{\eps}{2}\bigg((A \oplus B) \oplus -(A \oplus B)\bigg)
     ~ = ~ (A \oplus B)_\eps,
\end{align*}
as desired.

Returning to the proof, if $A \cap B \neq \emptyset$ then by Lemma~\ref{lem:basic}(a), $O \in A \oplus (-B)$, and the approximate membership query returns ``yes,'' as desired. If $A_\eps \cap B_\eps = \emptyset$ then by Lemma~\ref{lem:basic}(a) we have $O \notin A_\eps \oplus -(B_\eps)$ and by the above identity and the easy fact that $-(B_\eps) = (-B)_\eps$, we have $O \notin (A \oplus (-B))_\eps$, implying that the approximate membership query returns ``no.''
\end{proof}

The previous lemma relates approximate polytope intersection with an approximate membership of the origin in a polytope (Figure~\ref{f:dual}(a)). Determining whether the origin lies within the convex hull of a set of points $S$ is a classic problem in computational geometry, which can be solved by linear programming. However, we are interested in a faster approximate solution that does not compute $S$ explicitly. We cannot afford to preprocess an approximate polytope membership data structure for $A \oplus (-B)$ for each pair $A$ and $B$, since the number of such pairs is quadratic in the number of input polytopes. Instead, we preprocess each input polytope individually, and we show next how to efficiently answer approximate polytope membership queries for $A \oplus (-B)$ by using augmented data structures for approximate directional width queries for $A$ and $B$ as black boxes.

\begin{figure}[tbp]
  \centerline{\includegraphics[scale=0.7]{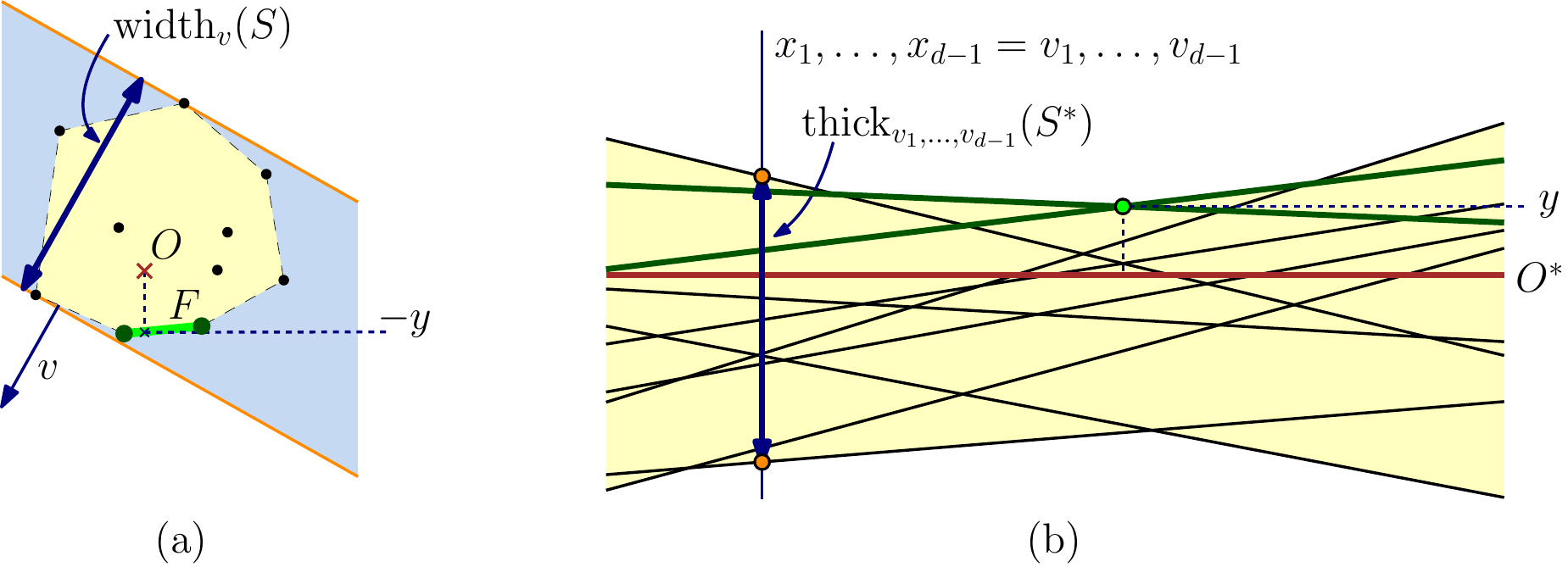}}
  \caption{(a) Primal problem of determining if $O \in \conv(S)$. (b) Dual problem of determining if the horizontal hyperplane $O^*$ is between the upper and lower envelopes.}
  \label{f:dual}
\end{figure}

\begin{lemma} \label{lem:widthtomembership}
Given augmented data structures for answering $\eps$-approximate directional width queries for polytopes $A$ and $B$, we can answer $\eps$-approximate membership queries for $A \oplus (-B)$ using $O(\polylog\inv\eps)$ queries to these data structures.
\end{lemma}

\begin{proof}
Without loss of generality, we may translate space so that the query point coincides with the origin $O$. Let $K = A \oplus (-B)$, and let $S$ be $K$'s vertex set. (Note that $K$ and $S$ are not explicitly computed.)

The problem of determining whether $O \in K$ is invariant to scaling and rotation about the origin. It will be helpful to perform some affine transformations that will guarantee certain properties for $K$. First, we apply Lemma~\ref{lem:fatten} to fatten $K$ and then apply a uniform scaling about the origin so that $K$'s diameter is $\Theta(1)$. By fatness, $K$ has a $\lambda$-sandwiching ball of radius $r = \Theta(1)$. If the origin either lies within the inner ball or outside the outer ball, then the answer is trivial. Otherwise, let $\Delta = 2 \lambda r$ be the diameter of the outer ball. We may apply a rotation about the origin so that the center of this ball lies on the positive $x_d$ axis at a point $(0,\ldots,0,\beta)$.
Again, this scaling and rotation can be computed in constant time using the augmented information. It follows that the coordinates of the points of $S$ have absolute values at most $\Delta = \Theta(1)$.

In summary, there exists an affine transformation computable in constant time such that after applying this transformation, the query point lies at the origin, $K = \conv(S)$ is sandwiched between two concentric balls of constant radii centered at $c = (0,\ldots,0,\beta)$, where $0 < \beta \le \Delta = O(1)$, and $K$'s vertex set $S$ is contained within $[-\Delta,\Delta]^d$. It is an immediate consequence that $\width_v(K) = \Theta(1)$ for all directions $v$, and hence it suffices to answer the membership query to an absolute error of $\Theta(\eps)$.

Lemma~\ref{lem:basic}(c) implies that we can answer $\eps$-approximate width queries for $K$ as the sum of two $\eps$-approximate width queries to $A$ and $B$. Therefore, our goal is to determine approximately if $O \in K$ using only approximate width queries to $A$ and $B$. In order to do this, we look at the projective dual problem in which each point $p = (p_1,\ldots,p_d) \in S$ is mapped to the hyperplane $p^*: x_d = p_1 x_1 + \cdots + p_{d-1}x_{d-1} - p_d$. Let $S^*$ denote the corresponding set of hyperplanes. The primal problem $O \in K$ is equivalent to the dual problem of determining whether the horizontal hyperplane $O^*: x_d = 0$ is sandwiched between the upper and lower envelopes of $S^*$ (Figure~\ref{f:dual}(b)). Since the point $c$ lies vertically above the origin and within $K$'s interior, it follows that $O^*$ cannot intersect the lower envelope. Therefore, it suffices to test whether $O^*$ intersects the upper envelope.

The dual problem can be solved exactly by computing the minimum value $y$ of the $x_d$-coordinate in the upper envelope and testing whether $y > 0$. In the primal, the value of $y$ corresponds to the negated $x_d$-coordinate of the intersection of a facet $F$ of the lower convex hull of $K$ and a vertical line passing through the origin (see Figure~\ref{f:dual}). Let $F$'s supporting hyperplane be denoted by $x_d = w_1 x_1 + \cdots + w_{d-1}x_{d-1} - w_d$. Since $K$ is sandwiched between two concentric balls of constant radii whose common center lies on this vertical line, it follows from simple geometry that this supporting hyperplane cannot be very steep. In particular, there exists $\alpha = O(1)$ such that $w_i \in [-\alpha,\alpha]$, for $i = 1,\ldots,d-1$. In the dual, this means that the minimum value $y$ is attained at a point whose first $d-1$ coordinates all lie within $[-\alpha,\alpha]$. In approximating $y$, we will apply directional width queries only for directional vectors $v = (v_1, \ldots, v_d)$ whose first $d-1$ coordinates lie within $[-\alpha,\alpha]$ and $v_d = -1$. Thus, $\|v\| = O(1)$. 

By Lemma~\ref{lem:dual}, the duals of two points $p, q \in S$ returned by an exact directional width query $\width_v(K)$ in the primal for a vector $v = (v_1,\ldots,v_{d-1},-1)$ correspond to the two dual hyperplanes in the upper and lower envelopes of $S^*$ that intersect the vertical line $x_i = v_i$ for $i = 1,\ldots,d-1$. Since queries are only applied to directions $v$ where $\|v\| = O(1)$ and since $\width_v(K) = \Theta(1)$ for all directions $v$, it follows from Lemma~\ref{lem:dual} that a relative error of $\eps$ in the directional width implies an absolute error of $O(\eps)$ in the corresponding thickness. We can think of the upper envelope of $S^*$ as defining the graph of a convex function over the domain $[-\alpha,\alpha]^{d-1}$. Since $S \subset [-\Delta,\Delta]^d$, the slopes of the hyperplanes in $S^*$ are similarly bounded, and therefore this function has bounded slope. It follows that, for an appropriate $\eps' = \Theta(\eps)$, we can compute this function to an absolute error of $\eps$ at any $(v_1,\ldots,v_{d-1})$ by performing an $(\eps')$-approximate directional width query on $K$ for $v = (v_1,\ldots,v_{d-1},-1)$. To complete the proof, it suffices to show that with $O(\polylog\inv\eps)$ such queries, it is possible to compute an absolute $\eps$-approximation to $y$. We do this in the next section.
\end{proof}

\subsection{Convex Minimization} \label{ss:minimization}

The following lemma shows how to use binary search to solve a one-dimensional convex minimization problem approximately (see Figure~\ref{f:binarysearch}(a)).

\begin{figure}[tbp]
  \centerline{\includegraphics[scale=0.7]{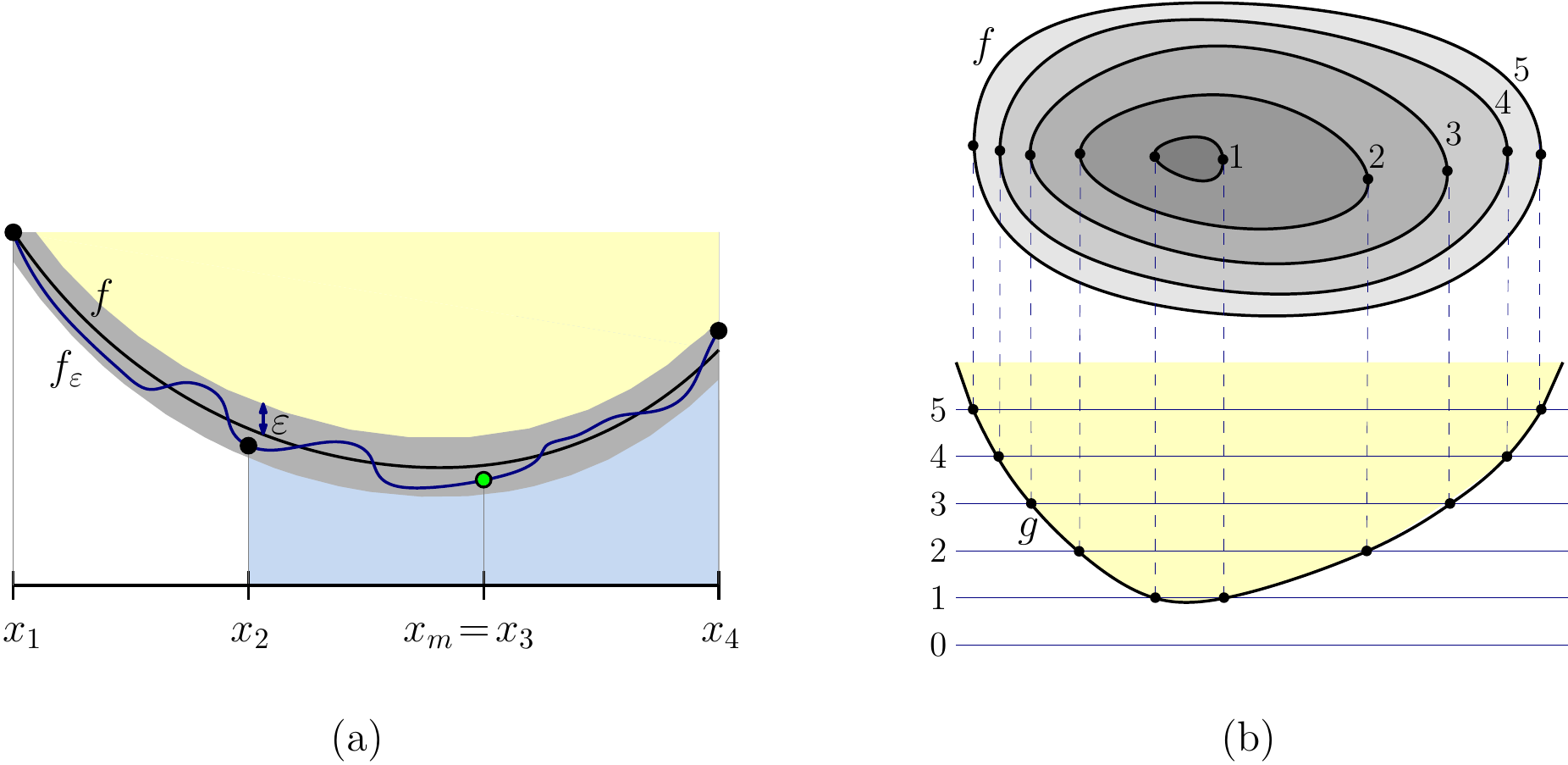}}
  \caption{(a) One-dimensional convex minimization. (b) Higher-dimensional convex minimization.}
  \label{f:binarysearch}
\end{figure}

\begin{lemma} \label{lem:bin-1d}
Let $a,b \in \RE$ and $\eps \in \RE+$ be real parameters. Let $f : [a,b] \rightarrow \RE$ be a convex function with bounded slope and $f_\eps : [a,b] \rightarrow \RE$ be a function with $|f(x) - f_\eps(x)| \leq \eps$ for all $x \in [a,b]$. Let $x^* \in [a,b]$ be the value of $x$ that minimizes $f(x)$. It is possible to determine a value $x'$ with $f(x') - f(x^*) = O(\eps)$ after $O(\log ((b-a)/\eps))$ evaluations of $f_\eps(\cdot)$ and no evaluation of $f(\cdot)$.
\end{lemma}

\begin{proof}
First, we present the recursive algorithm used to determine the value $x'$.
If $b-a < \eps$, then since the function has bounded slope, we simply return $x' = a$, as a valid answer.

Otherwise, we start by trisecting the interval $[a,b]$ and evaluate $f_\eps(x)$ at the four endpoints $x_1,x_2,x_3,x_4$ of the subintervals (see Figure~\ref{f:binarysearch}(a)).
Let $m$ denote the value $i$ that minimizes $f_\eps(x_i)$, breaking ties arbitrarily.
To simplify the boundary cases, let $x_0 = a$ and $x_5 = b$.
We then invoke our algorithm recursively on the interval $[x_{m-1}, x_{m+1}]$ and store the value returned as $x''$. We return the value $x$ among the two values $x_m,x''$ that minimizes $f_\eps(x)$.

Since the length of the interval reduces by at least one third at each iteration, the number of recursive calls and therefore evaluations of $f_\eps(\cdot)$ is $O(\log ((b-a)/\eps))$. 
Next, we show that $f(x') - f(x^*) = O(\eps)$. By the convexity of $f$ we have
\[ f(x) ~\geq~ f(x_{m+1}) + 3 (x-x_{m+1}) (f(x_{m+1})-f(x_m))/(b-a)               \text{, for }x \geq x_{m+1}.\]
Using that $|f(x) - f_\eps(x)| \leq \eps$, we have
\[ f(x) ~\geq~ f_\eps(x_{m+1}) - \eps + 3 (x-x_{m+1}) (f_\eps(x_{m+1})-f_\eps(x_m)-2\eps)/(b-a)               \text{, for }x \geq x_{m+1}.\]
Since $f_\eps(x_m) \leq f_\eps(x_{m+1})$, we have 
\[ f(x) ~\geq~ f_\eps(x_m) - \eps - 6 \eps (x-x_{m+1}) /(b-a)               \text{, for }x \geq x_{m+1}.\]
For $x$ inside the interval $[a,b]$ we have $|x-x_{m+1}| \leq b-a$, and therefore
\[ f(x) ~\geq~ f_\eps(x_m) - 7 \eps                 \text{, for }x_{m+1} \leq x \leq b.\]
The same argument is used to bound the case of $a \leq x \leq x_{m-1}$, obtaining
\[ f(x) ~\geq~ f_\eps(x_m) - 7 \eps                 \text{, for }x \notin [x_{m-1},x_{m+1}].\]

Either the minimum of $f(x)$ is inside the interval $[x_{m-1},x_{m+1}]$ or not. If it is not, then the previous inequality shows that $f_\eps(x_m)$ provides a good approximation, regardless of the value returned in the recursive call. If the minimum is inside the interval $[x_{m-1},x_{m+1}]$, then the recursive call will provide a value result by an inductive argument.
\end{proof}

We are now ready to extend the result to arbitrary dimensions.

\begin{lemma} \label{lem:bin-kd}
Let $a,b \in \RE$ and $\eps \in \RE+$ be real parameters. Let $f : [a,b]^d \rightarrow \RE$ for a constant dimension $d$ be a convex function with bounded slope and $f_\eps : [a,b]^d \rightarrow \RE$ be a function with $|f(x) - f_\eps(x)| \leq \eps$ for all $x \in [a,b]^d$. Let $x^* \in [a,b]^d$ be the value of $x$ that minimizes $f(x)$. It is possible to determine a value $x'$ with $f(x') - f(x^*) = O(\eps)$ after $O(\log^d ((b-a)/\eps))$ evaluations of $f_\eps(\cdot)$ and no evaluation of $f(\cdot)$.
\end{lemma}
\begin{proof}
The minimum $f(x^*)$ can be written as
\[f(x^*) = \min_{x \in [a,b]^d} f(x) = \min_{x_1 \in [a,b]} \; \min_{\tilde{x} \in [a,b]^{d-1}} f(x_1,\tilde{x}).\]
Note that if $f(x)$ is a convex function with bounded slope, then so is the function $g : [a,b] \rightarrow \RE$ (see Figure~\ref{f:binarysearch}(b)) defined as
\[g(x_1) = \min_{\tilde{x} \in [a,b]^{d-1}} f(x_1,\tilde{x}).\]

The proof is based on induction on the dimension $d$. Since $d$ is a constant, the number of induction steps is also a constant. The base case of $d=1$ follows from Lemma~\ref{lem:bin-1d}. By the induction hypothesis, we can solve the $(d-1)$-dimensional instance to obtain a function $g'(x_1)$ such that
\[|g(x_1) - g'(x_1)| = O(\eps).\]
Using Lemma~\ref{lem:bin-1d} for the function $g'(\cdot)$, we obtain a value $x'$ with  $f(x') - f(x^*) = O(\eps)$.

For the number of function evaluations $t(d)$ for a given dimension $d$ we have
\[t(1) = O(\log ((b-a)/\eps)) \text{ and}\]
\[t(k) = t(1) \cdot t(k-1).\]
The recurrence easily solves to the desired
\[t(d) = O(\log^d ((b-a)/\eps)).\qedhere\]
\end{proof}

By applying Lemma~\ref{lem:bin-kd} to the dual problem defined in the proof of Lemma~\ref{lem:widthtomembership} (where $f$ is the graph of the upper envelope of $S^*$ and $[a,b] = [-\alpha,\alpha]$) with the augmented data structure from Lemma~\ref{lem:widthqueries}, we obtain Theorem~\ref{thm:intersection} for the case when the input polytopes are represented by points. 
We will consider the case when the input polytopes are represented by halfspaces at the end of the next section.

\section{Minkowski Sum Approximation} \label{s:minkowski}

In this section, we will prove Theorems~\ref{thm:minkowski} and~\ref{thm:width}, as well as Theorem~\ref{thm:intersection} for the case when the input polytopes are represented by halfspaces. Assume that we are given two polytopes $A$ and $B$ in the point representation, and we have computed the augmented approximate directional width data structures from Lemma~\ref{lem:widthqueries} for each polytope. The objective is to obtain an $\eps$-approximation of the Minkowski sum $A \oplus B$ of size $O(1/\eps^{(d-1)/2})$ using these data structures.
Our approach is to fatten $A \oplus B$ using Lemma~\ref{lem:fatten} and then apply Dudley's construction~\cite{Dud74} in order to obtain an approximation with $O(1/\eps^{(d-1)/2})$ halfspaces. For completeness, we start by describing Dudley's algorithm.

Let $K \subset [-1,1]^d$ be a fat polytope of constant diameter.
Dudley's algorithm obtains an $\eps$-approximation represented by halfspaces as follows. Let $D$ be a ball of radius $2\sqrt{d}$ centered at the origin. (Note that $K \subset D$.) Place a set $W$ of $\Theta(1/\eps^{(d-1)/2})$ points on the surface of $D$ such that every point on the surface of $D$ is within distance $O(\sqrt{\eps})$ of some point in $W$. For each point $w \in W$, let $w'$ be its nearest point on the boundary of $K$. We call these points \emph{samples}. For each sample point $w'$, take the supporting halfspace passing through $w'$ that is orthogonal to the vector from $w'$ to $w$. The approximation is defined as the intersection of these halfspaces (see Figure~\ref{f:dudley}(a)).

\begin{figure}[tbp]
  \centerline{\includegraphics[scale=0.7]{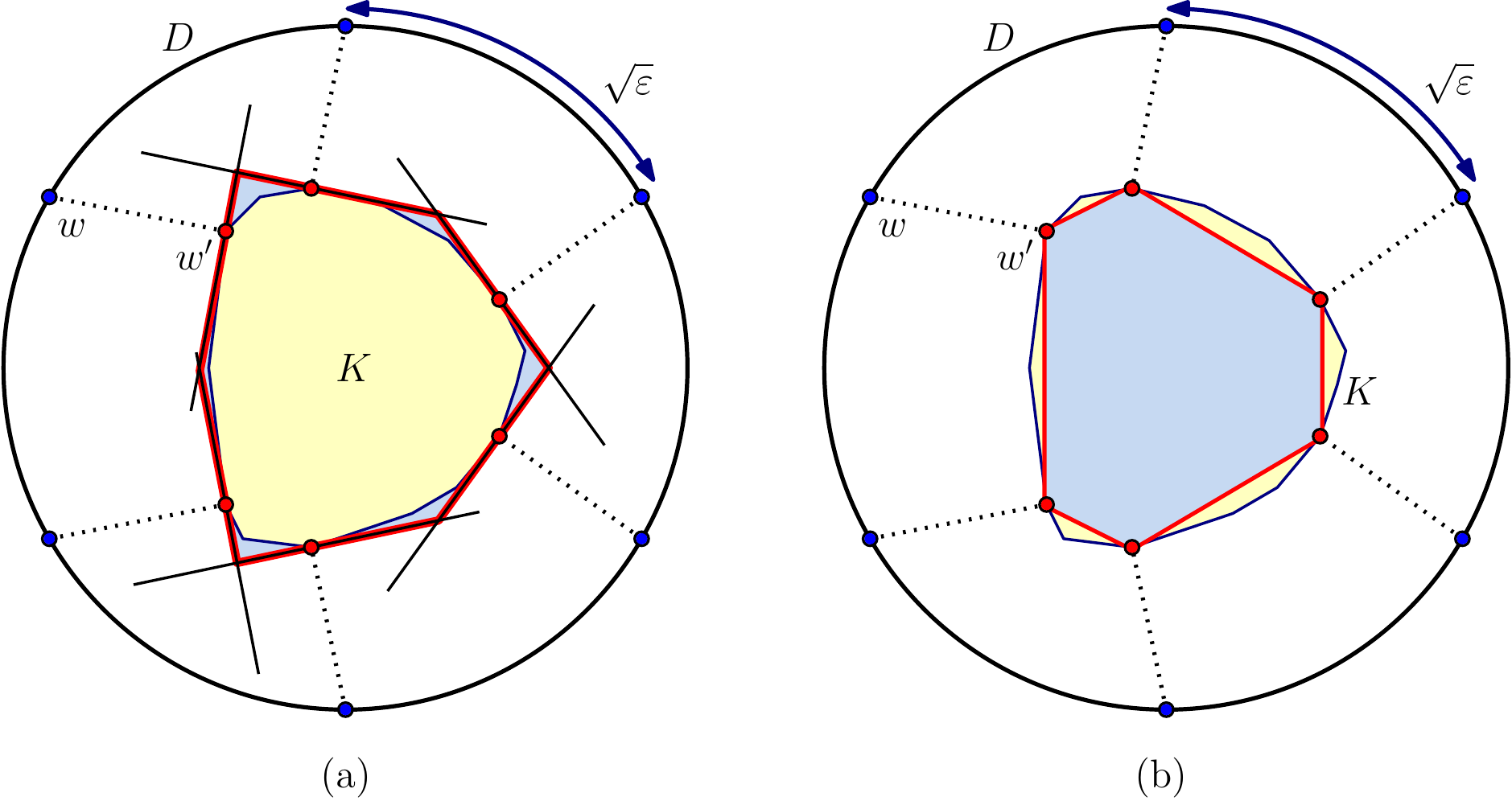}}
  \caption{(a) Dudley's and (b) Bronshteyn and Ivanov's polytope approximations.}
  \label{f:dudley}
\end{figure}

Bronshteyn and Ivanov~\cite{BrI76} presented a similar construction. Instead of approximating $K$ by halfspaces, Bronshteyn and Ivanov's construction approximates $K$ as the convex hull of the aforementioned set of samples%
\footnote{Dudley's construction yields an outer approximation and Bronshteyn and Ivanov's yields inner approximation, but it is possible to convert both to the other type through standard techniques. For details, see Lemma 2.8 of the full version of \cite{AFM17b}.}
(see Figure~\ref{f:dudley}(b)). In both constructions it is possible to tune the constant factors so that closest point queries need only be computed to within an absolute error of $\Theta(\eps)$.

An approximate closest point query between a polytope $K$ and a point $p$ within constant distance from $K$ can be reduced to computing an $\eps$-approximation to the smallest radius ball centered at $p$ that intersects $K$. This can be solved through binary search on the radius of this ball, where each probe involves determining whether $K$ intersects a ball of some radius centered at $p$. Notice that the data structure for approximate polytope intersection from Section~\ref{s:intersection} only accesses the bodies through approximate directional width queries, besides the initial fattening transformation. By Lemma~\ref{lem:basic}(c), given two preprocessed bodies $A$ and $B$, we can answer directional width queries on $A \oplus B$ through directional width queries on $A$ and $B$ individually. (In the case of a ball, no data structure is required.) Therefore, we can test intersection with a Minkowski sum $A \oplus B$, as long as we have augmented approximate directional width data structures for both $A$ and $B$.

In order to establish Theorem~\ref{thm:minkowski} for the case when the input polytopes are represented by points, we apply the aforementioned binary search to simulate Dudley's construction. Each sample is obtained after $O(\log \inv \eps)$ $\eps$-approximate polytope intersection queries. The total running time is dominated by the preprocessing time of Lemma~\ref{lem:widthqueries}. Note that the output polytope may be represented by either points or halfspaces according to whether we use Dudley's or Bronshteyn and Ivanov's algorithm.
To show that the input polytopes may be represented by halfspaces, we show how to efficiently convert between the two representations.

\begin{lemma}
Given an approximation parameter $\eps > 0$ and a polytope $K \subset \RE^d$ of size $n$ (given either using a point or halfspace representation), we can obtain an $\eps$-approximation of size $O(1/\eps^{(d-1)/2})$ (in either representation, independent of the input representation) in $O(n \log\inv\eps + 1/\eps^{(d-1)/2+\alpha})$ time, where $\alpha > 0$ is an arbitrarily small constant.
\end{lemma}

\begin{proof}
The case when the input is represented by points is a trivial case of Theorem~\ref{thm:minkowski}, where $B = \{O\}$. For the alternative case, it suffices to obtain an $\eps$-approximation of the polar polytope after fattening. (For details see Lemma 2.9 of the full version of~\cite{AFM17b}.)
\end{proof}

We remind the reader that Agarwal \etal~\cite{AGHRS00} showed that the width of a convex body $K$ is equal to the minimum distance from the origin to the boundary of the convex body $K \oplus (-K)$. To obtain Theorem~\ref{thm:width}, we compute Dudley's approximation of $K \oplus (-K)$ and then we determine the closest point to the origin among the $O(1/\eps^{(d-1)/2})$ bounding hyperplanes of the approximation.


\pdfbookmark[1]{References}{s:ref}
\bibliographystyle{abbrv}
\bibliography{shortcuts,width}

\end{document}